\documentclass{article}
\usepackage{amsmath}
\usepackage{amssymb}
\usepackage{amsthm,mathtools,amsmath,amssymb,booktabs,subcaption}
\usepackage{hyperref}
\usepackage{theoremref}
\usepackage{comment}

\def\uball{\mathbb{B}}

\def\reals{\mathbb{R}}
\def\ereals{\overline{\mathbb{R}}}

\def\comp{\raise 1pt \hbox{$\scriptstyle\circ$}}

\def\argmin{\mathop{\rm argmin}\limits}
\def\argmax{\mathop{\rm argmax}\limits}
\def\minimize{\mathop{\rm minimize}\limits}
\def\maximize{\mathop{\rm maximize}\limits}
\def\st{\mathop{\rm subject\ to}}
\def\underto#1{\stackrel{#1}{\rightarrow}}
\def\dom{\mathop{\rm dom}}

\def\gph{\mathop{\rm gph}}
\def\lev{\mathop{\rm lev}\nolimits}

\def\upto{{\raise 1pt \hbox{$\scriptstyle \,\nearrow\,$}}}
\def\downto{{\raise 1pt \hbox{$\scriptstyle \,\searrow\,$}}}
\def\cl{\mathop{\rm cl}\nolimits}

\def\inte{\mathop{\rm int}\nolimits}
\def\rinte{\mathop{\rm rint}}

\def\tos{\rightrightarrows}

\def\ovr{\mathop{\rm over}}

\newtheorem{theorem}{Theorem}

\newtheorem{lemma}[theorem]{Lemma}

\newtheorem{remark}[theorem]{Remark}
\newtheorem{assumption}{Assumption}

\title{Social welfare and price discovery in double auction markets}
\author{Teemu Pennanen\\ {\scriptsize King's College London}
}

\begin{document}

\maketitle

\begin{abstract}
The tendency of the double auction mechanism to drive prices to competitive equilibrium has been well documented in laboratory experiments, but the phenomenon has lacked a theoretical explanation. This paper studies dynamic double auctions in a pure exchange economy where agents bid their indifference prices implied by their current holdings and preferences. We show that Walras equilibria coincide with the fixed points of the double auction and that repeated double auctions generate bounded sequences of allocations and prices whose cluster points are Walras equilibria with transfers.
\end{abstract}

\section{Introduction}

Most securities, commodity and energy exchanges are based on the double auction mechanism, where buyers and sellers submit limit orders and the most generous offers on the two sides are matched. The laboratory experiments of Smith~\cite{smi62} showed that double auctions converge to the competitive equilibrium with remarkable reliability, even with few traders who know nothing of the aggregate supply and demand. Smith later framed these findings as support for the ``Hayek hypothesis'' that markets achieve efficient outcomes with dispersed private information and without price-taking or central coordination~\cite{smi82}. The finding has been replicated many times, most recently in a large-scale study by Ikica et al.~\cite{ijnnp23}, but its theoretical explanation has remained elusive; see the discussion e.g.\ in~\cite{plo08,smi10,prt13}.

What kind of theory could explain the efficiency of the double auction mechanism? Many of the existing models of double auctions explain price convergence through specific assumptions on how traders form their bids, such as equilibrium strategies or rules based on the history of bids and prices; see e.g.~\cite{wil87,fri91,el93,gd98,abl20}. Convergence is then attributed to the sophistication of the traders rather than to the market mechanism. Gode and Sunder~\cite{gs93} found in simulations that even traders who bid randomly, subject only to not trading at a loss, realize most of the available gains from trade, suggesting that much of the work is done by the market institution itself. Their results are, however, numerical, limited to single units, and do not address prices. Ikica et al.~\cite{ijnnp23} find experimentally that convergence is driven not by centralized price information but by centralized clearing.

In~\cite{pen22}, we proposed a model of double auctions in a pure exchange economy, where each agent bids their indifference prices given their current holdings and preferences. It was shown in~\cite{pen22} that the double auction clears with zero trades exactly at Pareto optimal allocations and that, starting from disequilibrium, repeated auctions generate allocations converging to individually rational Pareto optima. The experimental literature, however, is primarily concerned with prices: its central stylized fact is that transaction prices converge to competitive equilibrium prices. Unlike in the Walrasian model, where agents respond to announced prices, prices in a double auction are an output of the mechanism: by the time a transaction price is observed, the opportunity to trade at it has already passed. The present paper studies how such prices come to agree with competitive equilibrium prices.

Our main result, Theorem~\ref{thm:conv}, shows that every cluster point of the sequence of allocations and market-clearing prices generated by a double auction market is a Walras equilibrium with transfers. We also give conditions on the utilities under which repeated double auctions are well defined and generate bounded sequences of allocations and prices, even in the face of nonuniqueness in market-clearing, which is common in double auctions (Theorem~\ref{thm:bdd}). In the process, we show that an allocation-price pair is a Walras equilibrium if and only if the allocation is a double auction equilibrium and, up to normalization, the price is a corresponding market-clearing price (Theorem~\ref{thm:ft} and Remark~\ref{rem:cheaper}). As corollaries, we obtain versions of the first and second fundamental theorems of welfare economics without requiring that the utility functions be continuous. All our results apply to general concave utility functions without requiring monotonicity nor that their domains are contained in the positive orthant. This allows, in particular, for short positions and negative prices, which occur in practice, for example, in electricity markets.

Our results are related more broadly to the literature on trading out of equilibrium. In the Edgeworth process of Uzawa~\cite{uza62} and the Hahn process of \cite{hn62}, trades take place out of equilibrium, but prices are either absent or still adjusted by the auctioneer, while in \cite{fis83}, stability rests on assumptions on the agents' expectations. 
In these models, trades follow a postulated adjustment rule. The recent article of Deride, Jofr\'e and Rockafellar~\cite{djr26}, in contrast, proposes a realistic model for dynamic bilateral trading and shows that it reaches an equilibrium of prices and holdings. In the context of financial markets, our model is concerned with trading of standardized products in exchanges, while the model of \cite{djr26} describes decentralized “over-the-counter” markets.


The rest of the paper is organized as follows. Section~2 recalls the double
auction and its variational formulation. Section~3 describes bidding based on
endowments and preferences, and Section~4 reviews properties of the resulting
indifference price functions. Section~5 relates double auction equilibria to
Walras equilibria and Pareto optimality. Section~6 establishes existence and
boundedness in repeated double auctions, while Sections~7 and~8 prove
convergence of allocations and prices, respectively. The appendix reviews
the relevant theory from convex optimization used in the analysis.

\section{The double auction}\label{sec:da}

In a double auction market, the most most generous buying offers are matched against the most generous selling offers. At the end of the bidding period, the selling offers can be arrange into a nondecreasing {\em supply curve} $s:\reals_+\to\reals$ that, for each $x\in\reals_+$, gives the marginal price if one were to buy $x$ units from the most generous sellers. Similarly, the buying offers can be arrange into a nondincreasing {\em demand curve} $d:\reals_+\to\reals$ that, for each $x\in\reals_+$, gives the marginal price if one were to sel $x$ units to the most generous buyers. The market is cleared by finding crossing point $x$. The values of the curves give the {\em market-clearing price} and the $x$ units are traded at that price. Figure~\ref{fig:double-auction} illustrates a situation where vertical parts of the curves overlap so that the market clearing price is not unique. It may also happen that the horizontal parts overlap so that the market clearing trades are not unique although, in that case, it is natural to match the maximal quantity. The conclusions of this paper do not depend on the rule used to choose market-clearing prices and trades in the case of nonuniqueness.

\begin{figure}[t]
  \centering
  \begin{subfigure}[t]{0.49\textwidth}
    \centering
    \includegraphics[width=\linewidth]{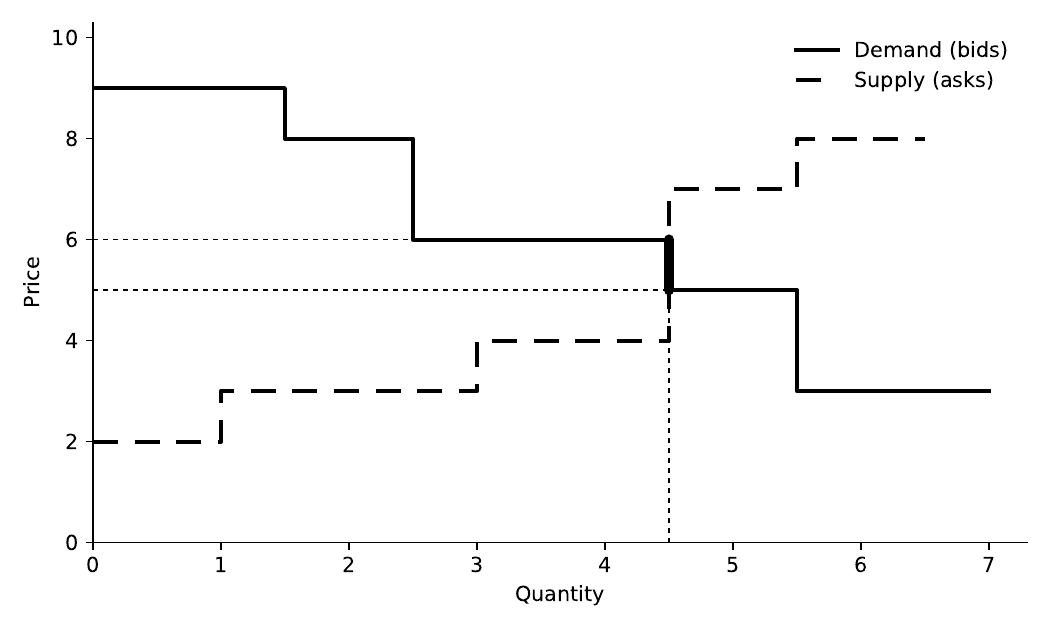}
    \caption{Demand and supply curves $d$ and $s$.}
    \label{fig:curves}
  \end{subfigure}\hfill
  \begin{subfigure}[t]{0.49\textwidth}
    \centering
    \includegraphics[width=\linewidth]{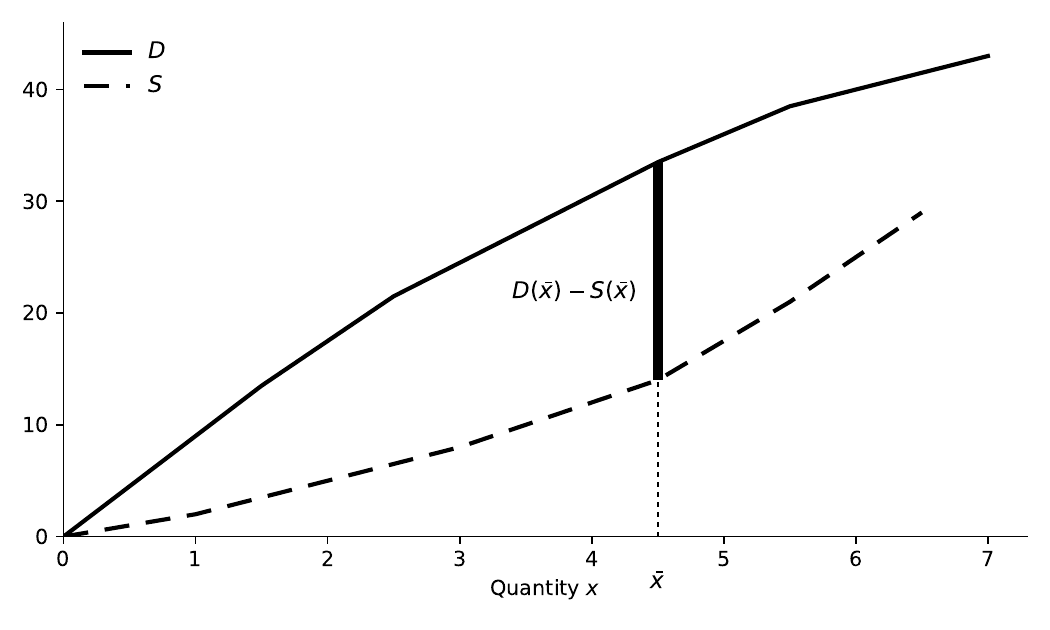}
    \caption{Revenue and cost functions $D$ and $S$.}
    \label{fig:surplus}
  \end{subfigure}
  \caption{Market clearing in a double auction. The market clearing
    quantity $\bar x = 4.5$ maximizes the consumer surplus $D(x)-S(x)$,
    and the market clearing prices form the interval
    $\partial D(\bar x)\cap\partial S(\bar x) = [5,6]$.}
  \label{fig:double-auction}
\end{figure}

Equivalently, defining
\[
S(x)=\int_0^xs(z)dz\quad\text{and}\quad D(x)=\int_0^xd(z)dz,
\]
the market clearing problem can be written as the optimization problem
\[
\maximize\quad D(x)-S(x)\quad\ovr x\in\reals_+,
\]
where the objective would give the maximum profit one could get buy buying $x$ units from the sellers and selling them to the willing buyers. Since $d$ is nonincreasing and $s$ is nondecreasing, the functions $D$ and $S$ are concave and convex, respectively. Assuming that the buying offers came from a finite set $I$ of buyers, we have
\[
D(x) = \sup\{\sum_{i\in I}D_i(x_i)\mid\sum_{i\in I}x_i=x\},
\]
where $D_i$ is the demand function submitted by agent $i\in I$. This simply means that the total amount $x$ sold is split among the potential buyers in order to maximize revenue of the sale. Similarly,
\[
S(x) = \inf\{\sum_{i\in I}S_i(x_i)\mid\sum_{i\in I}x_i=x\}.
\]
Interpreting sales as negative purchases and negative cost as revenue, the market clearing problem can be written as
\begin{equation}\label{da0}
\begin{aligned}
  &\maximize\quad & &\sum_{i\in I}D_i(x_i)\quad\ovr\ x\in\reals^I,\\
  &\st\quad & &\sum_{i\in I}x_i=0,
\end{aligned}
\end{equation}
where $x:=(x_i)_{i\in I}$ and we have extended the definition of the demand function by setting
\[
D_i(x_i) 
=
\begin{cases}
D_i(x_i) & \text{if $x_i\ge 0$},\\
-S_i(-x_i) & \text{if $x_i\le 0$}.
\end{cases}
\]
The least selling price submitted by a rational agent is at least as high as his highest buying price, so the function $D_i$ is concave. In most double auctions in practice, the traders can only submit ``limit orders'', which correspond to simple polyhedral functions $D_i$. More general concave functions $D_i$ can be approximated by finite collections of limit orders; see \cite[Section~2]{pen22}.

As in \cite{lue92,lue96,pen22}, we will study a market with multiple assets. We do this in order to draw connections to the classical welfare theory and also because the usual single-asset double auction is covered as a special case. We will assume that the market has a finite set $J$ of asset and that the traders submit demand function $D_i:\reals^J\to\ereals$. The market clearing problem now reads
\begin{equation}\label{da}
\begin{aligned}
  &\maximize\quad & &\sum_{i\in I}D_i(x_i)\quad\ovr\ x\in\reals^{I\times J},\\
  &\st\quad & &\sum_{i\in I}x_i=0,
\end{aligned}
\end{equation}
where the only difference to \eqref{da0} is that each $x_i$ is a portfolio vector in $\reals^J$. It should be noted that we allow the functions $D_i$ to take values in the {\em extended reals} $\ereals:=\reals\cup\{-\infty,+\infty\}$. This allows us to describe agents who only submit e.g.\ buy orders by defining $D_i(x)=-\infty$ for $x\notin\reals^J_+$.

Problem \eqref{da} is amenable to convex analysis; see~\cite{roc70a} and \cite{pen22}. The {\em Lagrangian function} associated with \eqref{da} is given by
\[
L(x,p) = \sum_{i\in I}[D_i(x_i) - x_i\cdot p]
\]
for all $x\in\reals^{I\times J}$ and $p\in\reals^J$. Under quite general conditions, a $z\in\reals^{I\times J}$ solves the market clearing problem \eqref{da} if and only if there exists a price vector $p$ such that
\begin{equation}\label{kkt}
\begin{aligned}
  \partial D_i(z_i)&\ni p,\\
  \sum_{i\in I} z_i&=0;
\end{aligned}
\end{equation}
see Theorem~\ref{thm:kkt2} in the appendix. Here and in what follows, $\partial D_i(z_i)$ denotes the {\em superdifferential} of $D_i$ at $z_i$ i.e., the set of {\em supergradients} $v_i\in\reals^J$ of $D_i$ at $z_i$ defined by
\begin{equation}\label{eq:sg}
D_i(x_i)\le D_i(z_i) + (x_i-z_i)\cdot v_i\quad\forall x_i\in\reals^J.
\end{equation}
Optimal solutions to \eqref{da} are called {\em market clearing trades} while the vector $p$ gives the {\em market-clearing prices}. At market clearing, agent $i\in I$ buys the portfolio $z_i$ for price $z_i\cdot p$.

Since $D_i(0)=0$, choosing $x_i=0$ in \eqref{eq:sg} gives
\[
z_i\cdot p\le D_i(z_i).
\]
In other words, the agent pays less for $z_i$ than what they offered to pay for it. The difference
\[
CS_i:=D_i(z_i)-z_i\cdot p
\]
is called the {\em consumer surplus}. The constraints of \eqref{da} imply that the optimum value in \eqref{da} is equal to the  {\em total consumer surplus}
\begin{equation}\label{eq:cs}
CS:=\sum_{i\in I}CS_i.
\end{equation}

If the market is cleared by the zero vector $z=0$, the double auction is said to be in {\em equilibrium}. The corresponding $\bar p$ gives the {\em equilibrium prices}. It is natural to ask if that kind of equilibria have something to do with classical notions of equilibria in exchange economies where agents' preferences over different commodity bundles are modeled by utility functions. We will answer this question in Section~\ref{sec:equilibrium}.

The market-clearing prices solve the {\em dual problem}
\begin{equation}\label{dual}
\minimize\quad -\sum_{i\in I}D_i^*(p)\quad\ovr\ p\in\reals^J,
\end{equation}
where
\begin{equation}\label{eq:conj}
D_i^*(p) := \inf_{x_i\in\reals^J}\{x_i\cdot p - D_i(x_i)\}
\end{equation}
is the {\em concave conjugate} of $D_i$. Indeed, conditions \eqref{kkt} hold if and only if $x$ solves~\eqref{da}, $p$ solves \eqref{dual} and the optimum values are equal; see Theorem~\ref{thm:kkt1} in the appendix. One can thus study market-clearing prices by analyzing the dual problem~\eqref{dual}.

It is worth noting that we allow for negative positions as well as negative prices. However, if even one agent has demand function which is nondecreasing in the partial order given by the positive orthant $\reals^J_+$, then the prices are necessarily positive. A more general property is given by the following remark. As usual, a set $C$ is said to be a {\em cone} if $x\in C$ and $\alpha>0$ imply $\alpha x\in C$.

\begin{remark}\label{lem:cone}
If $D_i$ is nondecreasing with respect to a cone $K_i\subset\reals^J$ in the sense that $x-x'\in K_i$ implies
\[
D_i(x)\le D_i(x'),
\]
then the market-clearing prices belong to the {\em polar cone}
\[
K_i^*:=\{p\in\reals^J\mid x\cdot p\le 0\ \forall x\in K_i\}.
\]
In particular, if $K_i=\reals^J_-$, then $p\in\reals^J_+$.
\end{remark}

\begin{proof}
The definition in \eqref{eq:sg} gives
\[
D_i(z_i+x_i)\le D_i(z_i) + x_i\cdot p\quad\forall x_i\in -K_i.
\]
Under the growth property, this implies
\[
0\le x_i\cdot p\quad\forall x_i\in -K_i.
\]
In other words, $p\in K_i^*$.
\end{proof}

\section{Bidding with endowments and preferences}

In general, the agents' trading offers depend on their preferences as well as their existing holdings in the assets. For example, an agent with a large short position in a risky asset may be willing to buy the asset at a price higher than an agent who has a large long position in the same asset. Such dependencies are naturally described by the agents' ``indifference prices'' as defined below.

We describe the preferences of agent $i\in I$ by a {\em utility function} $u_i:\reals^J\to\ereals$. Allowing the utility functions to take on the value $-\infty$, we can describe constraints an agent may have on their holdings. The set
\[
\dom u_i:=\{x_i\in\reals^J\mid u_i(x_i)>-\infty\}
\]
is known as the {\em domain} of $u_i$ or the {\em feasible set} of agent $i\in I$. Also, we do not require the domains of the functions to be contained in $\reals^J_+$ nor that the functions be increasing with respect to $\reals^J_+$. This allows, in particular, the agents to take short positions. Note also that we do not assume the utility functions to be differentiable. This allows for defining $u_i$ through optimization which is relevant for agents e.g.\ with {\em production facilities} which may be controlled by the agent.

We will denote the {\em initial holdings} of agent $i\in I$ by $x^0_i\in\reals^J$. We will assume throughout that $x^0_i\in\dom u_i$. The {\em indifference buying price} of agent $i\in I$ for a portfolio $x_i\in\reals^J$ is given by
\begin{equation}\label{idp}
D_i(x^0_i;x_i):=\sup\{r\in\reals\mid u_i(x_i^0+x_i-rg)\ge u_i(x_i^0)\},
\end{equation}
Here $g\in\reals^J$ is a {\em numeraire portfolio} in terms of which all prices are quoted. In most double auctions, prices are quoted in terms of cash. This corresponds to $g$ being the unit vector with one unit of cash and zero in all other assets. The indifference price $D_i(x^0_i;x_i)$ is the highest price the agent could pay for $x_i$ without worsening their holdings as measured by their utility function. In other words, the agents only offer to enter trades that would make them better off than they were before the trade.

Indifference pricing has a long history in finance and insurance going back, at least, to \cite{hn89}. The indifference pricing principle can be found already in~\cite[Section~4.1.2]{buh70} under the name of ``principle of zero utility"; see \cite{car9} and \cite{pen14} for further references. Indifference pricing can also be seen as an extension of ``risk-neutral'' pricing to which it reduces in the case of replicable instruments; see~\cite[Theorem~4.1]{pen14}. The indifference functions are closely related to the ``benefit functions'' studied in \cite{lue92,lue96} and defined as in \eqref{idp} but with $x^0=0$ on the left side of the inequality.`

The dependence of the indifferences price on the current holdings $x^0_i$ is quite natural from the practical point of view. For example, an agent who has a large position in a risky asset, may have less appetite for buying more of the asset than someone who has a short position in the asset. The dependence on the current holdings will be the key to explaining the tendency of double auction markets to a converge to an equilibruium; see Sections~\ref{sec:ra} and \ref{sec:prices} below. We keep the agent's current portfolio $x^0_i$ in the notation explicitly to emphasize the dependence.

We assume that each agent bids their indifference prices. This is a natural
benchmark for several reasons. Bidding the indifference prices requires no
knowledge of the other agents' bids or preferences, and it is safe: whatever
the other agents bid, the resulting trade leaves the agent at least as well
off as before, since the price paid never exceeds the agent's indifference
price. Moreover, in markets with many participants, a single agent has little
influence on the market clearing price, and the incentive to deviate from
truthful bidding vanishes as the number of traders grows~\cite{rsw94,cs06}.
With few traders, an agent may gain by shading their bids, i.e.\ by
submitting a demand function that lies below their indifference function,
and in repeated auctions an agent may also spread trades over time to reduce
their price impact. Studies of such strategic behaviour in dynamic markets
suggest, however, that it slows down convergence to efficient allocations
rather than preventing it~\cite{vay99,dz17}. Truthful bidding thus isolates
the role of the market mechanism itself, and even in this setting a
theoretical explanation of price discovery in double auctions has been
missing.

Assuming all agents bid their indifference prices, the market clearing problem \eqref{da} can be written as
\[
\begin{aligned}
  &\maximize\quad & &\sum_{i\in I}D_i(x^0_i;x_i)\quad\ovr\ x\in\reals^{I\times J},\\
  &\st\quad & &\sum_{i\in I}x_i=0.
\end{aligned}
\]
After market clearing, the trades $z_i$ take place at market-clearing prices $p\in\reals^J$ so that an agent's existing position $x^0_i$ is changed to
\[
x^1_i := x^0_i + z_i - (z_i\cdot p)g.
\]
Since $z$ satisfies the constraints of \eqref{da}, we have
\[
\sum_{i\in I}x^1_i = \sum_{i\in I}x^0_i,
\]
so $x^1$ is a feasible reallocation of the initial holdings $x^0$.

The above model was introduced in~\cite{pen22}, where it was found to share many of the properties of double auctions observed in practice and in laboratory experiments. In particular, it was shown in~\cite[Section~6]{pen22} that, under mild conditions on the agents' preferences, repeated double auctions converge to Pareto optimal allocations. 
The main results of the present paper make further connections to classical welfare theory and show, in particular, that the cluster points of the market-clearing prices are competitive equilibrium prices, in the sense of classical welfare theory.

Plugging in the definitions of the indifference functions $D_i(x^0_i;\cdot)$ into \eqref{da} one can write the market clearing problem as
\begin{equation}\label{eq:opt}
\begin{aligned}
  &\maximize\quad & & r\quad\ovr\ r\in\reals,\ w\in\reals^{I\times J}\\
  &\st\quad & & \sum_{i\in I}w_i + rg = \sum_{i\in I}x^0_i,\\
  & & & u_i(w_i) \ge u_i(x_i^0)\quad i\in I;
\end{aligned}
\end{equation}
see \cite[Lemma~2]{pen22}. Thus, the total consumer surplus in \eqref{eq:cs} equals the greatest amount of the numeraire that can be removed from the economy while reallocating the assets so that no agent is worse off compared to the current allocation. Formulation \eqref{eq:opt} will be particularly useful in the analysis of repeated double auctions in Sections~\ref{sec:ebdd}, \ref{sec:ra} and \ref{sec:prices} below. The optimum value of problem \eqref{eq:opt} equals the ``resource benefit function'' defined in \cite[Section~4]{lue96} for the utility levels $u_i(x^0_i)$. The same object appears already in \cite{all43} where it is called `distributable surplus''; see~\cite[page~174]{all78}. 
The connection to the double auction mechanism seems to have gone unnoticed untill~\cite{pen22}.

To emphasize the dependence on the current allocation $x^0$, we will denote the optimum value of the market clearing problem by $CS(x^0)$.

\begin{remark}[Pareto improvements and diminishing total surplus]\label{rem:cs}
As observed in \cite[Remark~2]{pen22}, the double auction makes {\em Pareto improvements} in the sense that $u_i(x^1_i)\ge u_i(x^0_i)$ for all $i\in I$. Our standing assumption $x^0_i\in\dom u_i$ thus implies $x^1_i\in\dom u_i$. It follows from the formulation \eqref{eq:opt} of the market clearing problem that $CS(x^1)\le CS(x^0)$, i.e.\ the double auction reduces the remaining total consumer surplus. 
\end{remark}

\section{Gradients and growth properties}\label{sec:gd}

This section reviews some basic properties of the indifference functions $D_i$ which will be needed in the analysis of the double auction market later on. As in \cite{pen22}, we will assume that the numeraire is a strictly preferred asset in the following sense.

\begin{assumption}\label{ass:g}
The utility functions $u_i$ are upper semicontinuous, concave and strictly increasing in the direction of the numeraire asset $g$ in the sense that
\[
u_i(x_i+rg)>u_i(x_i)
\]
for all $x_i\in\dom u_i$ and $r>0$.
\end{assumption}

Trivially,
\[
D_i(x^0_i;x_i)=\sup\{r\in\reals\mid x_i-rg\in C_i(x_i^0)\},
\]
where
\[
C_i(x_i^0):=\{w\in\reals^J\mid u_i(x_i^0+w_i)\ge u_i(x_i^0)\}
\]
is the set of all portfolios the agent would be willing to add to their current holdings.

Much like the classical existence conditions in convex optimization (see \cite[Section~27]{roc70a}), our analysis of well-posedness of the double auction mechanism in Section~\ref{sec:ebdd} below is based on the theory of ``recession cones'' (``horizon cones'' in \cite{rw98}). Given a closed convex set $C\subset\reals^J$, its {\em recession cone}
\[
C^\infty:=\{x\in\reals^J\mid \bar x+\alpha x\in C,\ \forall\alpha>0\}
\]
is closed, convex and independent of the choice of $\bar x\in C$; see \cite[Theorem~8.3]{roc70a}. The recession cone describes the behaviour of the set infinitely far from the origin. In particular, the set $C$ is bounded if and only if $C^\infty=\{0\}$; see \cite[Theorem~8.4]{roc70a}. Given a proper usc concave function $f$, its {\em recession function}
\[
f^\infty(x):=\inf_{\alpha>0}\frac{f(\bar x+\alpha x)-f(\bar x)}{\alpha}
\]
is usc concave and positively homogeneous and independent of the choice of $\bar x\in\dom f$; see \cite[Theorem~8.5]{roc70a}. The recession function describes the behaviour of the function infinitely far from the origin. It is easily checked that the recession function of the indicator function $\delta_C$ of a set $C\subseteq\reals^n$ equals the indicator function of the recession cone $C^\infty$. Recall that the {\em indicator function} takes the value $0$ on $C$ and $+\infty$ outside of $C$. Given $\alpha<\sup f$, the recession cone of the {\em level set}
\[
\lev_\alpha f:=\{x\in\reals^J\mid f(x)\ge\alpha\}
\]
equals $\lev_0 f^\infty$; see \cite[Theorem~8.5]{roc70a}. 

We will denote the recession function of $D_i(x^0_i;\cdot)$ by $D_i^\infty(x^0_i;\cdot)$. 

\begin{lemma}\label{lem:D}
Let $x^0_i\in\dom u_i$. Under Assumption~\ref{ass:g}, the function $D_i(x^0_i;\cdot)$ is proper concave and upper semicontinuous with $D_i(x^0_i;0)=0$,
\[
D_i(x^0_i;x_i+\alpha g) = D_i(x^0_i;x_i)+\alpha\quad\forall\alpha\in\reals,
\]
\[
D^\infty_i(x^0_i;x_i)=\sup\{r\in\reals\mid x_i-rg\in C_i(x^0_i)^\infty\},
\]
and the supremum above and in the definition of $D_i$ are attained whenever they are finite. Moreover, $C_i(x^0_i)^\infty = \lev_0 u_i^\infty$.
\end{lemma}

\begin{proof}
We have
\[
-D_i(x^0_i;x)=\inf_{r\in\reals} f(r,x),
\]
where $f:\reals\times\reals^J\to\ereals$ is the convex function given by
\[
f_i(r,x_i):=-r+\delta_{C_i(x^0_i)}(x_i-rg).
\]
By \cite[Theorem~9.3]{roc70a},
\[
f_i^\infty(r,x_i):=-r+\delta_{C_i(x^0_i)^\infty}(x_i-rg).
\]
The concavity, upper semicontinuity and the expression for the recession function as well as the attainment of the infima now follow from \cite[Theorem~A.77]{pp24}. Since Assumption~\ref{ass:g} implies $u_i(x^0_i)<\sup u_i$, the last claim follows from \cite[Theorem~8.7]{roc70a}.
\end{proof}

Recall that the {\em support function} $\sigma_C:\reals^J\to\ereals$ of a set $C\subset\reals^J$ is the positively homogeneous convex function defined, for every $p\in\reals^J$, by
\[
\sigma_C(p):=\sup_{x\in C}x\cdot p.
\]
The {\em normal cone} of $C$ at a point $\bar x\in C$ is the closed convex cone
\[
N_C(\bar x):=\{p\in\reals^J\mid \bar x\cdot p\ge x\cdot p\quad\forall x\in C\}.
\]
We will denote the conjugate and the superdifferential of the function $D_i(x^0_i;\cdot)$ by $D_i^*(x^0_i;\cdot)$ and $\partial D_i(x^0_i;x_i)$, respectively; see \eqref{eq:conj} and \eqref{eq:sg}.

\begin{lemma}\label{lem:D2}
The conjugate of $D_i(x^0_i;\cdot)$ can be expressed as
\[
D_i^*(x^0_i;p) =
\begin{cases}
  -\sigma_{C_i(x_i^0)}(-p) & \text{if $g\cdot p=1$},\\
  -\infty & \text{otherwise}.
\end{cases}
\]
Under Assumption~\ref{ass:g},
\[
\partial D_i(x^0_i;x_i) = \{p\in\reals^J\mid -p\in N_{C_i(x_i^0)}(x_i-D_i(x^0_i;x_i)g),\ g\cdot p=1\}
\]
for any $x^0_i\in\dom u_i$. 
\end{lemma}

\begin{proof}
As to the conjugate,
\begin{align*}
D_i^*(x^0_i;p) &= \inf_{x_i}\{x_i\cdot p - D_i(x^0_i;x_i)\}\\
&= \inf_{x_i,r}\{x_i\cdot p - r\mid x_i-rg\in C_i(x_i^0)\}\\
&= \inf_{w_i,r}\{w_i\cdot p - r(1-g\cdot p)\mid w_i\in C_i(x_i^0)\},
\end{align*}
which equals $-\sigma_{C_i(x_i^0)}(-p)$ if $g\cdot p=1$ and $-\infty$ otherwise.

Let $x^0_i\in\dom u_i$. By definition of a supergradient, we have $p\in\partial D_i(x^0_i;x_i)$ if and only if
\[
D_i(x^0_i;x_i)+D_i^*(x^0_i;p)=x_i\cdot p.
\]
Using the expression for the conjugate, this can be written as
\begin{align*}
  D_i(x^0_i;x_i)-\sigma_{C_i(x_i^0)}(-p) &= x_i\cdot p,\label{eq:fe}\\
  g\cdot p &= 1.\nonumber
\end{align*}
Since $g\cdot p=1$, the first equation can be written as
\[
\sup_{x_i'\in C_i(x_i^0)}(x_i'- x_i+D_i(x^0_i;x_i)g)\cdot (-p)=0.
\]
Since, by Lemma~\ref{lem:D}, the supremum in the definition of $D_i$ is attained, we have $x_i-D_i(x^0_i;x_i)g \in C_i(x_i^0)$, so the previous equation means that
\[
(x_i'- x_i+D_i(x^0_i;x_i)g)\cdot (-p)\le 0\quad\forall x'\in C_i(x_i^0).
\]
In other words,
\[
-p\in N_{C_i(x_i^0)}(x_i-D_i(x^0_i;x_i)g),
\]
which completes the proof.
\end{proof}

\begin{remark}[Cash and redundant trades]
If $\sum_{i\in I}\alpha_i=0$, then by Lemma~\ref{lem:D}, the portfolios $z'_i:=z_i+\alpha_i g$ satisfy
\[
\sum_{i\in I}D_i(x^0_i;z'_i) = \sum_{i\in I}D_i(x^0_i;z_i)
\]
and
\[
\sum_{i\in I}z'_i=\sum_{i\in I}z_i.
\]
In particular, the market clearing trades in \eqref{da} are not unique. The portfolio updates are, however, independent of $\alpha$. Indeed, since $g\cdot p=1$, by Lemma~\ref{lem:D2},
\begin{align*}
x^1_i &= x^0_i + z'_i - (z'_i\cdot p)g\\
&= x^0_i + z_i+\alpha_ig - [z_i\cdot p+\alpha_i]g\\
&= x^0_i + z_i - (z_i\cdot p)g.
\end{align*}
We can eliminate the degeneracy, by requiring that the market clearing trades $z$ satisfy $g\cdot z_i=0$ for all $i\in I$. Indeed, if $z$ is feasible in \eqref{da}, then
\[
z'_i:=z_i-\frac{z_i\cdot g}{|g|^2}g
\]
is feasible, has the same objective value and satisfies $g\cdot z'_i=0$. Economically, the additional condition simply means that the agents won't submit offers to trade the numeraire. In particular, if the numeraire portfolio $g$ is the unit vector with one unit of cash and zero in all other assets, then $z'_{i,j}=z_{i,j}$ for all $j\in J$ except for cash for which $z'_{i,j}=0$. 
\end{remark}

By Lemma~\ref{lem:D2}, the dual problem \eqref{dual} can be written as
\begin{equation}\label{dual2}
\begin{aligned}
  &\minimize\quad & &\sum_{i\in I}\sigma_{C_i(x^0_i)}(-p) \quad\ovr\ p\in\reals^J\\
  &\st\quad & & p\cdot g=1.
\end{aligned}
\end{equation}
By Theorem~\ref{thm:kkt1}, the market-clearing prices are solutions of the dual problem. This will be used in Section~\ref{sec:ebdd} to establish the boundedness of market-clearing prices in repeated double auctions.

In general, our model allows for negative positions as well as negative prices. However, if even one agent has a utility function that is nondecreasing in positive directions then the prices are necessarily positive. More generally we have the following.

\begin{remark}
If $u_i$ is nondecreasing with respect to a cone $K_i\subset\reals^J$, then so too is $D_i$ so that, by Remark~\ref{lem:cone}, market-clearing prices satisfy $p\in K_i^*$.
\end{remark}

\section{Double auction equilibria and welfare theorems}\label{sec:equilibrium}

An {\em allocation-price pair} $(\bar x,\bar p)$ is said to be a {\em Walras equilibrium} (or a ``competitive equilibrium'') for the initial allocation $x^0$ if $\bar x$ is a feasible reallocation of $x^0$ and
\begin{equation}\label{eq:argmax}
\bar x_i\in\argmax_{x_i\in\reals^J}\{u_i(x_i)\,|\,\bar p\cdot x_i\le\bar p\cdot x^0_i\}
\end{equation}
for every agent $i\in I$. An allocation $\bar x\in\reals^{I\times J}$ is said to be {\em Pareto optimal} (or Pareto efficient) if there is no feasible reallocation $x$ of $\bar x$ such that $u_i(x_i)\ge u_i(\bar x_i)$ for all $i\in I$ with a strict inequality for at least one $i\in I$.

The 1st fundamental theorem of welfare economics gives conditions under which the allocation component of a Walras equilibrium is Pareto optimal. The 2nd theorem gives conditions under which a Pareto optimal allocation can be paired with a price vector to form a Walras equilibrium; see e.g.~\cite{var92,mwg95}. This section shows that Pareto optimality and Walras equilibria are intimately related to equilibria in double auction markets. We say that an allocation $\bar x\in\reals^{I\times J}$ is a {\em double auction equilibrium} if the double auction in \eqref{da} clears with zero trades. This happens precisely when the total consumer surplus is zero, i.e.
\[
CS(\bar x)=0;
\]
see \eqref{eq:cs} and Remark~\ref{rem:cs}. Interestingly, this is essentially the notion of ``equilibrium'' proposed by Allais, although he did not consider double auctions. The definition at the bottom of \cite[page~134]{all78} asks the surplus to be zero with respect to any asset as a numeraire; we have the fixed numeraire $g$.

By Theorem~\ref{thm:kkt1}, a sufficient condition for $CS(\bar x)=0$ is that there exist a vector $\bar p\in\reals^J$ such that
\[
\partial D_i(\bar x_i;0)\ni\bar p,
\]
for every agent $i\in I$, and then, $\bar p$ solves the dual problem \eqref{dual} or, equivalently, \eqref{dual2}. Indeed, by Theorem~\ref{thm:kkt1}, the above condition implies that $z=0$ solves the market clearing problem, so the optimum value is zero. The vector $\bar p$ gives the {\em market-clearing prices} for the double auction equilibrium.

Theorem~3 of \cite{pen22} gave sufficient conditions under which Pareto optimal allocations coincide with double auction equilibria. This section gives sufficient conditions under which an allocation-price pair $(\bar x,\bar p)$ is a Walras equilibrium if and only if $\bar x$ is a double auction equilibrium and, up to scaling, $\bar p$ is the corresponding market-clearing price; see Theorem~\ref{thm:ft} and Remark~\ref{rem:cheaper} below. The theorem also gives versions of the fundamental theorems of welfare economics, but without making the usual continuity assumption on the utility functions. Assumption~\ref{ass:g} on a strictly preferred numeraire asset allows us to get away with Assumptions~\ref{ass:4} and \ref{ass:uslater} below.

Instead of continuity of the utility functions, we will hold on to Assumption~\ref{ass:g} and the following, where the {\em interior} of a set $C\subset\reals^J$ is denoted by $\inte C$.

\begin{assumption}\label{ass:uslater}
\[
0\in\inte\sum_{i\in I}\dom D_i(\bar x_i;\cdot).
\]
\end{assumption}

Under Assumption~\ref{ass:g}, $C_i(\bar x_i) + \alpha g\subseteq C_i(\bar x_i)$ for all $\alpha>0$, so
\begin{equation}\label{eq:domD}
\dom D_i(\bar x_i;\cdot) = \bigcup_{\alpha\in\reals}[C_i(\bar x_i) + \alpha g] = \bigcup_{\alpha\le 0}[C_i(\bar x_i) + \alpha g].
\end{equation}

\begin{remark}
Assumption~\ref{ass:uslater} holds under Assumption~\ref{ass:g} if there is even one agent $i\in I$ whose utility function is continuous at $\bar x_i$.
\end{remark}

\begin{proof}
The continuity implies $\bar x_i\in\inte\dom u_i$, so there is an $\epsilon>0$ such that $\bar x_i+\epsilon g\in\inte\dom u_i$. Under Assumption~\ref{ass:g}, $u_i(\bar x_i+\epsilon g)> u_i(\bar x_i)$. By \cite[Theorem~10.1]{roc70a}, $u_i$ is continuous on $\inte\dom u_i$, so there is a neighbourhood $U$ of the origin such that
\[
u_i(\bar x_i+x_i+\epsilon g)\ge u_i(\bar x_i)\quad\forall x_i\in U.
\]
In other words, $U+\epsilon g\subset C_i(\bar x_i)$. By \eqref{eq:domD}, this implies $U\subset\dom D_i(\bar x_i;\cdot)$. Since $0\in\dom D_i(\bar x_i;\cdot)$ for all $i\in I$, this implies Assumption~\ref{ass:uslater}.
\end{proof}

The {\em relative interior} $\rinte C$ of a set $C$ is the interior of $C$ with respect to its affine hull; see \cite[Section~6]{roc70a}. The relative interior of a convex set is always nonempty, and it equals the interior whenever the latter is nonempty; see \cite[Theorem~6.2]{roc70a}. By \cite[Theorem~6.6]{roc70a}, Assumption~\ref{ass:uslater} thus implies
\begin{equation}\label{eq:slater}
0\in\inte\sum_{i\in I}\dom D_i(\bar x_i;\cdot) = \sum_{i\in I}\rinte\dom D_i(\bar x_i;\cdot),
\end{equation}
a condition that will be used in the proof of Theorem~\ref{thm:ft} below.

The following was introduced in \cite[Section~5]{pen22}.

\begin{assumption}\label{ass:4}
If $x_i$ is such that $u_i(x_i)>u_i(x_i')$ for some $x'_i\in\dom u_i$ then $x_i-\varepsilon g\in\dom u_i$ for small enough $\varepsilon>0$.
\end{assumption}

Assumption~\ref{ass:4} means that if a portfolio $x_i$ is not among the worst in an agent's feasible set $\dom u_i$, then the agent can give away a small amount of the numeraire and remain feasible. This can be seen as an extension of the ``indispensability'' condition of \cite{jrw23} where the domains of the utility functions were assumed to be contained in the positive orthant. Under Assumption~\ref{ass:4} a Walrasian ``quasi-equilibrium'' is a true Walrasian equilibrium; see e.g.~\cite[Section~16.D]{mwg95} and the last part of the proof of Theorem~\ref{thm:ft} below. Assumption~\ref{ass:4} allows also for a uniform version of the classical ``cheaper point'' condition; see Remark~\ref{rem:cheaper} below.

\begin{theorem}\label{thm:ft}
Given a feasible reallocation $\bar x$ of $x^0$ and a price vector $\bar p$, consider the following properties:
\begin{enumerate}
\item[(a)]
  $(\bar x,\bar p)$ is a Walras equilibrium for the initial allocation $\bar x$ (``equilibrium with transfers'');
\item[(b)]
  $\bar x$ is Pareto optimal;
\item[(c)]
  $\bar x$ is a double auction equilibrium;
\item[(d)]
  $\bar x$ is a double auction equilibrium with market-clearing price $\bar p$.
\end{enumerate}
Assume that $\bar x_i\in\dom u_i$ for all $i\in I$ and that Assumption~\ref{ass:g} holds. We have (a) $\Rightarrow$ (b) $\Rightarrow$ (c). Under Assumption~\ref{ass:uslater}, (c) implies the existence of a price vector $\bar p$ such that (d) holds. Under Assumption~\ref{ass:4}, (d) $\Rightarrow$ (a).
\end{theorem}

\begin{proof}
Assume (a). We have $\bar p\cdot x_i>\bar p\cdot\bar x_i$ for every $x_i$ with $u_i(x_i)>u_i(\bar x_i)$, while Assumption~\ref{ass:g} implies that $\bar p\cdot x_i\ge\bar p\cdot\bar x_i$ for every $x_i$ with $u_i(x_i)\ge u_i(\bar x_i)$. Indeed, if there was an $x_i$ with $\bar p\cdot x_i<\bar p\cdot\bar x_i$ and $u_i(x_i)\ge u_i(\bar x_i)$, then for small enough $\epsilon>0$ the portfolio $x'_i:=x_i+\epsilon g$ would satisfy $\bar p\cdot x_i'\le \bar p\cdot\bar x_i$ and $u_i(x_i')>u_i(\bar x_i)$, in contradiction to the assumption that $(\bar x,\bar p)$ be a Walras equilibrium. It follows that, if $x$ was such that $u_i(x_i)\ge u_i(\bar x_i)$ with a strict inequality for some $i\in I$, we would have
\[
\bar p\cdot\sum_{i\in I}x_i > \bar p\cdot\sum_{i\in I}\bar x_i
\]
so $x$ could not be feasible. Thus, (b) holds.

If $\bar x$ does not satisfy (c), then formulation \eqref{eq:opt} of the market clearing problem implies the existence of an $r>0$ such that the constraints of problem \eqref{eq:opt} are satisfied. We can then construct another feasible allocation by giving $rg$ to one of the agents. Under Assumption~\ref{ass:g}, this would lead to strict increase of the agent's utility so $\bar x$ can't be Pareto. Thus, (b) implies (c).

Assume (c). Since Assumption~\ref{ass:uslater} implies \eqref{eq:slater}, Theorem~\ref{thm:kkt2} gives the existence of a price vector $\bar p\in\reals^J$ such that
\begin{equation}\label{eq:dae}
\partial D_i(\bar x_i;0)\ni\bar p\quad \forall i\in I,
\end{equation}
which means that (d) holds.

Assume (d) so that \eqref{eq:dae} holds. Under Assumption~\ref{ass:g}, Lemma~\ref{lem:D} gives $D_i(\bar x_i;0)=0$ so, by Lemma~\ref{lem:D2}, condition \eqref{eq:dae} can be written as
\begin{align*}
  N_{C_i(\bar x_i)}(0)&\ni-\bar p\quad \forall i\in I,\\
  g\cdot\bar p &=1
\end{align*}
which means that
\begin{align*}
  x_i\cdot\bar p &\ge 0\quad\forall x_i\in C_i(\bar x_i),\quad \forall i\in I\\
  g\cdot\bar p &=1,
\end{align*}
where the first condition can be written more explicitly as
\begin{equation}\label{eq:argmin}
\bar x_i\in\argmin\{x_i\cdot\bar p\mid u_i(x_i)\ge u_i(\bar x_i)\}.
\end{equation}
Assumption~\ref{ass:4} now implies (a). Indeed, assume for contradiction, that there exists an $x_i\in\reals^J$ with $u_i(x_i)>u_i(\bar x_i)$ and $x_i\cdot\bar p\le\bar x_i\cdot\bar p$. Assumption~\ref{ass:4} then gives the existence of an $\varepsilon>0$ such that $x_i-\epsilon g\in\dom u_i$. Since $u_i$ is concave, we can find an $\varepsilon'\in(0,\epsilon)$ with $u_i(x_i-\varepsilon'g)$ arbitrarily close to $u_i(x_i)$ and, in particular with $u_i(x_i-\varepsilon'g)\ge u_i(\bar x_i)$. On the other hand, since $g\cdot\bar p=1$, we would get
\[
(x_i-\varepsilon'g)\cdot\bar p = x_i\cdot\bar p - \varepsilon' \le\bar x_i\cdot\bar p - \varepsilon',
\]
contradicting \eqref{eq:argmin}.
\end{proof}

The implication (a) $\Rightarrow$ (b) in Theorem~\ref{thm:ft} is known as the 1st fundamental theorem of welfare economics; see e.g.~\cite{var92,mwg95}. Inspecting the above proof, it is clear that the implication holds under the ``nonsatiability'' condition: {\em For any $x_i\in\dom u_i$ and $\epsilon>0$, there is an $x_i'\in\reals^J$ with $|x_i'-x_i|\le\epsilon$ and $u_i(x_i')>u_i(x_i)$.} This is even weaker than Assumption~\ref{ass:g} which simply means that the numeraire is strictly preferred by all.

The 2nd fundamental theorem, says that (b) implies the existence of a price vector $\bar p$ such that (a) holds.  It is usually proven under the assumption that the utility functions be continuous at the Pareto point $\bar x$; see e.g.~\cite[Section~17.7]{var92}. Theorem~\ref{thm:ft} gives the same implication under Assumptions~\ref{ass:g}, \ref{ass:4} and \ref{ass:uslater}. This does not require continuity. Lemma~\ref{lem:slater} below gives a condition that implies that Assumption~\ref{ass:uslater} holds for any $\bar x$.

Many treatments of the welfare theorems assume in addition that the utility functions are nondecreasing in the directions of the positive orthant $\reals^J_+$ and that the domains of the utility functions are contained in $\reals^J_+$. Such assumptions were not used above. This allows, in particular, for markets where the agents may take short positions and some assets may turn out to have negative prices. 

The implication (b) $\Rightarrow$ (c) in Theorem~\ref{thm:ft} was observed already in \cite[Theorem~3]{pen22}. That result shows also that the converse implication holds under Assumption~\ref{ass:4}.

The following gives the implication (a) $\Rightarrow$ (d) without invoking Assumptions~\ref{ass:4} or \ref{ass:uslater}. Instead, it uses an assumption that depends on the equilibrium prices~$\bar p$.

\begin{remark}[Cheaper point-condition]\label{rem:cheaper}
Let Assumption~\ref{ass:g} hold and let $(\bar x,\bar p)$ be a Walras equilibrium for the initial allocation $\bar x$. If, for every $i\in I$, there exist $x_i\in\dom u_i$ with $x_i\cdot\bar p<\bar x_i\cdot\bar p$, then $\bar x$ is a double auction equilibrium with market-clearing price
\[
\bar p':=\bar p/(g\cdot\bar p).
\]
Under Assumption~\ref{ass:4}, this condition holds independently of $\bar p$, if for every $i\in I$, there exists an $x'_i\in\dom u_i$ such that $u_i(\bar x_i)>u_i(x'_i)$.
\end{remark}

\begin{proof}
Indeed, Walras equilibrium means that each $\bar x_i$ solves
\begin{equation}\label{eq:proof}
\begin{aligned}
  &\maximize\quad & & u_i(x_i)\quad\ovr\ x_i\in\reals^J\\
  &\st\quad & & x_i\cdot\bar p\le \bar x_i\cdot\bar p.
\end{aligned}
\end{equation}
Assumption~\ref{ass:g} implies that we must have $g\cdot\bar p>0$. By scaling $\bar p$ if necessary, we can assume that $g\cdot\bar p=1$. By Theorem~\ref{thm:kkt2}, the second condition guarantees the existence of Lagrange multipliers $\lambda_i\ge 0$ such that
\[
\bar x_i\in\argmax_{x_i\in\reals^J}\{u_i(x_i)-\lambda_ix_i\cdot\bar p\}.
\]
Since the functions $u_i$ are strictly increasing in $g$, we must have $\lambda_i>0$, so the condition can be written as
\[
\bar x_i\in\argmax_{x_i\in\reals^J}\{\frac{1}{\lambda_i}u_i(x_i)-x_i\cdot\bar p\}.
\]
By Theorem~\ref{thm:kkt1}, this implies that $\bar x$ satisfies \eqref{eq:argmin}. As observed in the proof of Theorem~\ref{thm:ft}, this means that $\bar x$ is a double auction equilibrium with market-clearing price $\bar p'$.

As to the last claim, assume that $u_i(\bar x_i)>u_i(x'_i)$. Assumption~\ref{ass:4} gives the existence of an $\epsilon>0$ such that $\bar x_i-\epsilon g\in\dom u_i$. Under Assumption~\ref{ass:g}, one must have $g\cdot\bar p>0$ at a Walras equilibrium $(\bar x,\bar p)$, so $(\bar x_i-\epsilon g)\cdot\bar p<\bar x_i\cdot\bar p$. 
\end{proof}

\section{Dynamic double auction markets}\label{sec:ebdd}

In general, the agent's demand functions $D_i$ depend on their endowments $x^0_i$. After market clearing, the endowments are changed to $x^1_i = x^0_i + z^0_i - (p\cdot z^0_i)g$ so the demands may change too. There is no reason for the new allocation $x^1$ to be a double auction equilibrium, in general.

From now on, we will study a dynamic double auction market where, each market clearing is followed by another bidding period and market clearing. We will denote the market clearing trades in the $t$th auction by $z^t\in\reals^{I\times J}$, the corresponding market-clearing prices by $p^t\in\reals^J$ and the resulting allocations by $x^t=(x^t_i)_{i\in I}$. Thus, when they exist, the sequences satisfy
\begin{equation}\label{eq:dasd}
\begin{aligned}
  \partial D_i(x^t_i;z^t_i)&\ni p^t,\\
  \sum_{i\in I} z^t_i&=0,\\
  x^{t+1}_i &= x^t_i + z^t_i - (z^t_i\cdot p^t)g.
\end{aligned}
\end{equation}
Note that the double auction does not define the sequences uniquely since the market clearing trades and prices need not be unique; see Figure~\ref{fig:double-auction} for a simple example. It is also conceivable that the market ends up in an allocation $x^t$ at which the next market clearing problem does not have market clearing trades or prices. This section gives sufficient conditions which guarantee that, no matter which market clearing trades $z^t$ and prices $p^t$ are chosen at the $t$th auction, the resulting allocation $x^{t+1}$ is such that, in the next market clearing problem, both $z^{t+1}$ and $p^{t+1}$ exist.

\begin{remark}[Continuous vs.\ call auctions]\label{rem:cda}
  When the market clearing frequency is high enough, the repeated double auction market studied here approximates continuous-time double auctions implemented in most securities exchanges. Indeed, orders submitted by agents who are not involved in market-clearing trades remain valid in the next round, since their holdings, and hence their indifference prices, do not change. It can be shown that the same applies to agents whose orders are executed at their indifference prices. If the market is cleared frequently enough, one may assume, as in continuous trading, that each incoming order is immediately matched or entered in the limit order book. If the size of an incoming order doesn't exceed the quantity available at the best price levels, the dynamics coincide with those of a continuous-time double auction market. 
Continuous double auctions will be addrerssed in future work.
\end{remark}

Recall from Section~\ref{sec:gd} that the market-clearing prices $p^t$ in \eqref{eq:dasd} solve the dual problem \eqref{dual2}. This observation will be employed below to show that, when defined, the sequence $(p^t)_{t=0}^\infty$ of market-clearing prices remain bounded. We will use the following.

\begin{assumption}\label{ass:bdd}
The utility functions $u_i$ satisfy
\begin{gather}
u_i^\infty(x_i)\ge 0,\ \sum_{i\in I}x_i=0\implies x=0,\label{ass:primalbdd}\\
g\in\inte\sum_{i\in I}\lev_0 u_i^\infty.\label{ass:dualbdd}
\end{gather}
\end{assumption}

The first condition in Assumption~\ref{ass:bdd} holds, in particular, if $u_i^\infty(x_i)<0$ for $x_i\notin\reals^J_+$. This, as well as condition \eqref{ass:dualbdd}, holds e.g.\ for the Cobb-Douglas utilities when $g\in\inte\reals^J_+$.

We assume throughout that $x^0_i\in\dom u_i$ for every agent $i\in I$. If \eqref{eq:dasd} holds, then $x^t_i\in\dom u_i$ for all $t$ since, by Remark~\ref{rem:cs}, the double auction makes Pareto improvements. 

\begin{lemma}\label{lem:slater}
Under Assumption~\ref{ass:g}, condition~\ref{ass:dualbdd} implies that Assumption~\ref{ass:uslater} holds for every $\bar x$ with $\bar x_i\in\dom u_i$ for all $i\in I$.
\end{lemma}

\begin{proof}
Let $\bar x$ be such that $\bar x_i\in\dom u_i$ for all $i\in I$. Since $0\in C_i(\bar x_i)$, we have $C_i(\bar x_i)^\infty\subseteq C_i(\bar x_i)$. Since Assumption~\ref{ass:g} implies $u_i(\bar x_i)<\sup u_i$, we have $C_i(\bar x_i)^\infty=\lev_0u^\infty_i$, by \cite[Theorem~8.7]{roc70a}, so $C_i(\bar x_i)\supseteq \lev_0u^\infty_i$. Condition \eqref{ass:dualbdd} thus implies
\[
g\in\inte\sum_{i\in I}C_i(\bar x_i).
\]
By equation \eqref{eq:domD}, this implies Assumption~\ref{ass:uslater}.
\end{proof}

By the above lemma, condition \ref{ass:dualbdd} allows us to invoke Theorem~\ref{thm:kkt2} on the existence of the price vector, independently of the current allocation $x^t$.



\begin{theorem}\label{thm:bdd}
Under Assumptions~\ref{ass:g} and \ref{ass:bdd}, the sequences of allocations and prices generated by repeated double auction exist and are bounded, independently of the choice of the market clearing trades and prices.
\end{theorem}

\begin{proof}
We prove the existence by induction. Assume that the price vector $p^{t-1}$ and the $t$th allocation $x^t$ is well-defined by the $t$th market clearing and corresponding exchange of the assets. By \cite[Theorem~27.1]{roc70a}, problem \eqref{da} admits a solution with $D_i=D_i(x^t_i;\cdot)$ provided the set
\[
C:=\{x\in\reals^{I\times J}\mid \sum_{i\in I}D_i^\infty(x^t_i;x_i)\ge 0,\ \sum_{i\in I}x_i=0\}
\]
is linear. By Lemma~\ref{lem:D},
\[
C:=\{x\in\reals^{I\times J}\mid \exists r_i\in\reals: \sum_{i\in I}r_i\ge 0,\ u_i^\infty(x_i-r_ig)\ge 0,\ \sum_{i\in I}x_i=0\}.
\]
Let $x\in C$ and $k\in I$. Defining, $x'_i:=x_i-r_ig$ for $i\ne k$ and $x'_k=x_k-r_kg+rg$, where $r:=\sum_{i\in I}r_i\ge 0$, we have
\[
u_i^\infty(x'_i)\ge 0,\ \sum_{i\in I}x'_i=0.
\]
Condition \eqref{ass:primalbdd} then forces $x'=0$ so that $x_i=r_ig$ for $i\ne k$ and $x_k=(r_k-r)g$. Since $x\in C$, we also have
\[
u_k^\infty(x_k-r_kg) = u_k^\infty(-rg)\ge 0. 
\]
Under Assumption~\ref{ass:g}, this can only hold if $r=0$. In summary,
\[
C:=\{x\in\reals^{I\times J}\mid \exists r_i\in\reals: \sum_{i\in I}r_i=0,\ x_i=r_ig\}.
\]
This set is clearly linear so problem \eqref{da} admits a solution $z^t$, by \cite[Theorem~27.1]{roc70a}.

Assumption~\ref{ass:uslater} holds, by Lemma~\ref{lem:slater}, so by Theorem~\ref{thm:kkt2}, dual solution $p^t$ exist and the optimum values of \eqref{da} and \eqref{dual} are equal. It now follows, by induction, that \eqref{eq:dasd} admits a solution $(x^t)_{t=0}^\infty$ and $(p^t)_{t=0}^\infty$ independently of the choice of the market clearing trades and prices.

It remains to show that the sequences are bounded. Since each $x^t$ is a feasible allocation and, by \cite[Remark~2]{pen22}, $u_i(x^{t+1}_i)\ge u_i(x^t_i)$ for all $i\in I$, the sequence $(x^t)_{t=0}^\infty$ is contained in the set
\[
C = \{x\in\reals^{I\times J}\mid \sum_{i\in I}x_i=\sum_{i\in I}x^0_i,\ u_i(x_i)\ge u_i(x_i^0)\ i\in I\}.
\]
By \cite[Theorem~8.4]{roc70a}, $C$ is bounded if and only if $C^\infty=\{0\}$. By \cite[Corollary~8.3.3]{roc70a} and \cite[Theorem~8.7]{roc70a},
\[
C^\infty = \{x\in\reals^{I\times J}\mid \sum_{i\in I}x_i=0,\ u_i^\infty(x_i)\ge 0\ i\in I\}
\]
which only contains the origin under Assumption~\ref{ass:bdd}.

As to the dual boundedness, we have $C_i(x^t_i)\supseteq \lev_0u^\infty_i$, as observed in the proof of Lemma~\ref{lem:slater}. The dual objective is thus minorized by
\[
\sigma(p):=\sum_{i\in I}\sigma_{\lev_0 u_i^\infty}(-p)
\]
independently of the allocation $x^t$. As already observed, the dual optimum value equals the optimum value $CS(x^t)$ of the market clearing problem. Since $CS(x^t)\le CS(x^0)$ (see Remark~\ref{rem:cs}), the sequence $(p^t)_{t=0}^\infty$ is contained in the set
\[
B:=\{p\in\reals^J\mid \sigma(p)\le CS(x^0),\ g\cdot p=1\}.
\]
By \cite[Theorem~8.4]{roc70a}, $B$ is bounded if and only if $B^\infty=\{0\}$. By \cite[Corollary~8.3.3]{roc70a} and \cite[Theorem~8.7]{roc70a},
\[
B^\infty=\{p\in\reals^J\mid \sigma(p)\le 0,\ g\cdot p=0\}.
\]
Since $0\in\lev_0 u_i^\infty$, we have $\sigma_{\lev_0 u_i^\infty}\ge 0$ so
\[
B^\infty=\{p\in\reals^J\mid \sigma_{\lev_0 u_i^\infty}(-p)\le 0\ \forall i\in I,\ g\cdot p=0\}.
\]
Since $\lev_0u_i^\infty$ is a cone, its support function equals the indicator function of the polar cone
\[
(\lev_0 u_i^\infty)^*:=\{p\in\reals^J\mid x_i\cdot p\le 0\ \forall x_i\in\lev_0u_i^\infty\},
\]
so
\[
B^\infty = \bigcap_{i\in I}(-\lev_0 u_i^\infty)^*\cap\{p\mid g\cdot p=0\}.
\]

Condition~\eqref{ass:dualbdd} implies that the set
\[
C:=\bigcup_{\alpha>0}\alpha(\sum_{i\in I}\lev_0 u_i^\infty-g) = \sum_{i\in I}\lev_0 u_i^\infty-\bigcup_{\alpha>0}\alpha g
\]
is the whole space $\reals^J$. Thus, the polar cone $C^\circ$ of $C$ reduces to the origin. A direct calculation gives
\[
C^\circ = \bigcap_{i\in I}(\lev_0 u_i^\infty)^*\cap\{p\mid g\cdot p\ge 0\},
\]
so $B^\infty\subseteq -C^\circ=\{0\}$.
\end{proof}

\section{Convergence to double auction equilibria}\label{sec:ra}

The following condition employed in \cite{pen22} is a quantitative strengthening of Assumption~\ref{ass:g}.

\begin{assumption}\label{ass:1}
The utility functions $u_i$ are upper semicontinuous, concave and, for every $r>0$ there exists a $\delta_i>0$ such that 
\[
u_i(x_i+r g)\ge u_i(x_i) + \delta_ir\quad\forall x_i\in\uball(r).
\]
\end{assumption}

Assumption~\ref{ass:1} is implied by Assumption~\ref{ass:g} if $u_i$ is finite everywhere, or more generally, if $x_i+rg\in\inte\dom u_i$ for all $x_i\in\cl\dom u_i$ and $r>0$. Concavity of $u_i$ implies that the difference quotient
\[
\frac{u_i(x_i+rg)-u_i(x_i)}{r}
\]
is nonincreasing in $r$ so the inequality in Assumption~\ref{ass:1} implies
\begin{equation}\label{eq:delta}
u_i(x_i+r'g)\ge u_i(x_i) + \delta_ir'\quad\forall r'\in(0,r].
\end{equation}

Recall that an allocation $\bar x$ is a double auction equilibrium if the market clears with zero trades, or equivalently, if the total consumer surplus $CS(x^t)$ is zero. As noticed in Remark~\ref{rem:cs}, $CS(x^t)$ is nonincreasing along the sequence $(x^t)_{t=0}^\infty$. An allocation $\bar x$ said to be {\em individually rational} if $u_i(\bar x_i)\ge u_i(x^0_i)$. The following was proved already in \cite{pen22}, but we repeat the proof to clarify some details that were left implicit in \cite{pen22}.

\begin{theorem}\label{thm:primalconv}
Let $(x^t,p^t)_{t=0}^\infty$ be a bounded sequence of allocation-price pairs generated by the repeated double auction market. Under Assumption~\ref{ass:1}, the total surplus $CS(x^t)$ is nonincreasing in $t$, there is a constant $\gamma\ge 0$ such that 
\[
CS(x^t)\le\frac{\gamma}{t}
\]
and all cluster points $\bar x$ of $(x^t)_{t=0}^\infty$ are individually rational double auction equilibria.
\end{theorem}

\begin{proof}
Recall that
\[
CS(x^t)=\sum_{i\in I}CS^t_i,
\]
where the agent $i$'s surplus $CS^t_i:=D_i(x^t_i;z_i^t)-p^t\cdot z^t_i$ is nonnegative. Since, by \cite[Remark~2]{pen22}, $u_i(x_i^t)$ are nondecreasing in $t$, the optimum value $CS(x^t)$ of problem \eqref{eq:opt} is nonincreasing. It follows that $CS^t_i$ are bounded. Since $x_i^t$ is bounded, the equation
\[
x_i^{t+1}=x_i^t + z_i^t-D_i(x^t_i;z_i^t)g + CS^t_ig
\]
implies that $z_i^t-D_i(x^t_i;z_i^t)g$ is bounded as well. Let $r>0$ be such that $x_i^t+z^t_i-D_i(x^t_i;z_i^t)g\in\uball(r)$ and $CS^t_i\le r$ for all $t$. Formulation \eqref{eq:delta} of Assumption~\ref{ass:1} and the definition of the indifference price $D_i(x^t_i;z_i^t)$ give
\[
u_i(x_i^{t+1})\ge u_i(x^t_i+z^t_i-D_i(x^t_i;z_i^t)g) + \delta_iCS^t_i \ge u_i(x_i^t) + \delta_iCS^t_i.
\]
Adding up over agents gives
\[
CS(x^t) \le \sum_{i\in I}\frac{u_i(x_i^{t+1})-u_i(x_i^t)}{\delta_i}.
\]
and adding up over iterations $s=0,\ldots,t-1$,
\[
\sum_{s=0}^{t-1}CS(x^s) \le \sum_{i\in I}\frac{u_i(x_i^t)-u_i(x_i^0)}{\delta_i}.
\]
Since $CS(x^t)$ is nonincreasing in $t$, the left side is greater than $tCS(x^t)$ so
\[
CS(x^t)\le\frac{1}{t}\sum_{i\in I}\frac{u_i(x_i^t)-u_i(x_i^0)}{\delta_i}.
\]
Since $(x^t)_{t=0}^\infty$ is bounded and $u_i$ are concave and usc, the values $u_i(x^t_i)$ are bounded from above. Thus, the bound in the statement holds with
\[
\gamma:=\sup_t\sum_{i\in I}\frac{u_i(x_i^t)-u_i(x_i^0)}{\delta_i}.
\]

If $\bar x$ is a cluster point of $(x^t)_{t=0}^\infty$, the upper semicontinuity of $u_i$ and the monotonicity of $u_i(x_i^t)$ in $t$ give 
\[
u_i(\bar x_i)\ge\limsup_{t\to\infty} u_i(x_i^t)=\sup_t u_i(x_i^t).
\]
Thus, $CS(\bar x)\le CS(x^t)$ for all $t$ so $CS(\bar x)=0$.
\end{proof}

\section{Convergence to equilibrium prices}\label{sec:prices}

Recall from \eqref{eq:dasd} that the sequences of allocations $x^t$ and prices $p^t$ generated by the double auction market satisfy
\begin{equation*}
\begin{aligned}
  \partial D_i(x^t_i;z^t_i)&\ni p^t,\\
  \sum_{i\in I} z^t_i&=0,\\
  x^{t+1}_i &= x^t_i + z^t_i - (z^t_i\cdot p^t)g.
\end{aligned}
\end{equation*}
Section~\ref{sec:ebdd} gave conditions under which the sequences exist and are bounded. The purpose of this section is to show that all the cluster points $(\bar x,\bar p)$ of the sequence $(x^t,p^t)_{t=0}^\infty$ satisfy
\begin{align*}
  \partial D_i(\bar x_i;0)&\ni \bar p
\end{align*}
which means that $\bar x$ is a double auction equilibrium and that $\bar p$ is the corresponding market-clearing price; see Section~\ref{sec:equilibrium}. Theorem~\ref{thm:ft} then implies that, under Assumptions~\ref{ass:g} and \ref{ass:4}, $\bar x$ is a Pareto optimal allocation and $(\bar x,\bar p)$ is a Walras equilibrium with transfers.

We start with two preparatory lemmas that involve the notion of set convergence in the Painlevé-Kuratowski sense; see \cite[Chapter~4]{rw98}. A sequence of sets $C^t\subset\reals^n$ is said to {\em converge} to another set $C$, written $C^t\to C$, if
\begin{enumerate}
\item
  for every $x\in C$, there exist a sequence $(x^t)_{t=1}^\infty$ with $x^t\in C^t$ and $x^t\to x$,
\item
  if $x$ is the limit of a subsequence of points $x^t\in C^t$, then $x\in C$.
\end{enumerate}

\begin{lemma}\label{lem:setc}
If $(x^t_i)_{t=0}^\infty$ is a sequence of portfolios converging to a point $\bar x_i\in\dom u_i$ such that $u_i(x^t_i)$ increases to $u_i(\bar x_i)$, then
\[
C_i(x_i^t)\to C_i(\bar x_i).
\]
\end{lemma}

\begin{proof}
By \cite[Exercise~4.3]{rw98}, $\lev_{u_i(x^t_i)}u_i\to\lev_{u_i(\bar x_i)}u_i$. Since
\[
C_i(x_i^t) = \{w\in\reals^J\mid u_i(x_i^t+w_i)\ge u_i(x_i^t)\} = \lev_{u_i(x^t_i)}u_i - x^t_i,
\]
the claim follows from \cite[Exercise~4.29]{rw98}.
\end{proof}

A sequence of set-valued mappings $S^t:\reals^n\tos\reals^m$ is said to {\em converge graphically} to another mapping $S:\reals^n\tos\reals^m$, written $S^t\underto{g}S$, if the sets
\[
\gph S^t:=\{(x,u)\in\reals^n\times\reals^m\mid u\in S^t(x)\}
\]
converge to
\[
\gph S:=\{(x,u)\in\reals^n\times\reals^m\mid u\in S(x)\}.
\]

\begin{lemma}\label{lem:gsd}
If $C_i(x_i^t)\to C_i(\bar x_i)$ and Assumption~\ref{ass:g} holds, then 
\[
\partial D_i(x_i^t;\cdot)\underto{g}\partial D_i(\bar x_i;\cdot).
\]
\end{lemma}

\begin{proof}
Note first that
\[
-D_i(x^t_i;x_i) = \inf_{r\in\reals}f^t(x_i,r),
\]
where
\[
f^t(x_i,r) := -r+\delta_{C_i(x^t_i)}(x_i-rg).
\]
The convergence $C_i(x^t_i)\to C_i(\bar x_i)$ implies the convergence of the epigraphs of their indicator functions. Thus, by \cite[Exercise~7.47]{rw98}, the epigraphs of the functions $f^t$ converge to the epigraph of
\[
f(x_i,r) := -r+\delta_{C_i(\bar x_i)}(x_i-rg)
\]
Thus, by \cite[Exercise~7.57]{rw98}, the epigraphs of the functions $D_i(x^t;\cdot)$ converge to $D_i(\bar x_i;\cdot)$ provided
\begin{equation}\label{eq:reccond}
f^\infty(0,r)>0\quad\forall r\ne 0.
\end{equation}
It is easily checked that the recession function of the indicator function of a convex set equals the indicator function of the recession cone of the set. Thus, by \cite[Exercise~3.29]{rw98},
\[
f^\infty(0,r) = -r+\delta_{C_i(\bar x_i)^\infty}(-rg),
\]
where $C_i(\bar x_i)^\infty=\lev_0 u_i^\infty$, by \cite[Theorem~8.7]{roc70a}. Assumption~\ref{ass:g} implies $-rg\notin\lev_0 u_i^\infty$ for all $r>0$, so \eqref{eq:reccond} holds. The claim now follows from the Attouch theorem, which says that epi-convergence of a sequence of convex functions implies the convergence of the graphs of the corresponding subdifferential mappings; see \cite[Theorem~12.35]{rw98}.
\end{proof}

Our main result is now a simple consequence of the already established properties. It says that all cluster points of the sequence of allocation-price pairs generated by the double auction equilibria and the corresponding market-clearing prices. The result assumes the boundedness of the sequences, which was established in Section~\ref{sec:ebdd}.

\begin{theorem}\label{thm:conv}
Let $(x^t,p^t)_{t=0}^\infty$ be a bounded sequence of allocation-price pairs generated by the repeated double auction market. If Assumption~\ref{ass:1} holds and the utility functions $u_i$ are continuous relative to their domains, then all cluster points of $(x^t,p^t)_{t=0}^\infty$ are double auction equilibria and corresponding market-clearing prices. If, in addition, Assumption~\ref{ass:4} holds, the cluster points are Walras equilibria with transfers and the allocations are Pareto efficient.
\end{theorem}

\begin{proof}
Defining $d^t_i:=z^t_i - (z^t_i\cdot p^t)g$, Lemma~\ref{lem:D2} gives $\partial D_i(x^t_i;z^t_i)=\partial D_i(x^t_i;d^t_i)$ and thus, \eqref{eq:dasd} can be written as
\begin{align*}
  \partial D_i(x^t_i;d^t_i)&\ni p^t,\\
  \sum_{i\in I} d^t_i&=0.
\end{align*}
Since $x^{t+1}_i=x^t_i+d^t_i$, the boundedness of $x^t$ implies the boundedness of $d^t$. It follows that if $(\bar x,\bar p)$ is a cluster point of the sequence of allocations and prices, we can find a subsequence along which $(x^t,d^t,p^t)$ converge to a point $(\bar x,\bar d,\bar p)$. Since, by Remark~\ref{rem:cs}, the utilities are nondecreasing in $t$ and since, by the standing assumption, $x^0_i\in\dom u_i$, we have $\lim_tu_i(x^t_i)>-\infty$. The upper semicontinuity of $u_i$ thus gives $u_i(\bar x_i)>-\infty$, so $\bar x_i\in\dom u_i$. The relative continuity of $u_i$ now gives $u_i(x^t_i)\upto u_i(\bar x_i)$. Assumption~\ref{ass:1} implies Assumption~\ref{ass:g}, so Lemmas~\ref{lem:setc} and \ref{lem:gsd}, give $\partial D_i(x_i^t;\cdot)\underto{g}\partial D_i(\bar x_i;\cdot)$. Thus,
\begin{align*}
  \partial D_i(\bar x_i;\bar d_i)&\ni\bar p,\\
  \sum_{i\in I}\bar d_i&=0.
\end{align*}
By Theorem~\ref{thm:kkt1}, this implies that $(\bar d,\bar p)$ are primal-dual optimal and that the optimum values of the primal and the dual are equal. On the other hand, Theorem~\ref{thm:primalconv} says that $CS(\bar x)=0$, so the zero trade is also primal optimal. By Theorem~\ref{thm:kkt1} again, $(0,\bar p)$ satisfies the optimality conditions. Thus, $\bar x$ is a double auction equilibrium and $\bar p$ is the corresponding market-clearing price. The last claim follows directly from Theorem~\ref{thm:ft}
\end{proof}

The convergence of the prices generated by double auctions has been observed in numerous experimental settings since the pioneering works of Smith~\cite{smi62}. The convergence was observed also in the numerical experiments of \cite{pen22}, but there, the phenomenon was left unexplained. The experiment involved five traded assets, the numeraire portfolio $g=(1,\ldots,0)$ and 100 agents with randomly generated initial endowments and Cobb-Douglas utilities. Table~\ref{tab:num} reproduces the progress of the total consumer surplus and the market-clearing prices $p^t$ along the iterated double auctions.

\begin{table}
\centering
{\scriptsize
\begin{tabular}{l|c|ccccccc}
\toprule
$t$ &     CS &     &     &    $p_t$ &     &     \\
\midrule
1  & 38.691 & 1.000 & 0.159 & 0.135 & 0.132 & 0.132 \\
2  & 11.534 & 1.000 & 0.690 & 0.585 & 0.580 & 0.577 \\
3  &  4.266 & 1.000 & 0.843 & 0.711 & 0.710 & 0.709 \\
4  &  2.091 & 1.000 & 0.894 & 0.754 & 0.753 & 0.753 \\
5  &  1.159 & 1.000 & 0.916 & 0.774 & 0.772 & 0.772 \\
6  &  0.711 & 1.000 & 0.928 & 0.784 & 0.782 & 0.782 \\
7  &  0.486 & 1.000 & 0.933 & 0.790 & 0.787 & 0.786 \\
8  &  0.362 & 1.000 & 0.936 & 0.792 & 0.789 & 0.789 \\
9  &  0.289 & 1.000 & 0.938 & 0.794 & 0.791 & 0.791 \\
10 &  0.236 & 1.000 & 0.939 & 0.795 & 0.792 & 0.792 \\
11 &  0.187 & 1.000 & 0.940 & 0.796 & 0.794 & 0.793 \\
12 &  0.141 & 1.000 & 0.941 & 0.797 & 0.795 & 0.795 \\
13 &  0.098 & 1.000 & 0.942 & 0.798 & 0.796 & 0.796 \\
14 &  0.061 & 1.000 & 0.943 & 0.799 & 0.797 & 0.797 \\
15 &  0.032 & 1.000 & 0.944 & 0.800 & 0.797 & 0.797 \\
16 &  0.012 & 1.000 & 0.944 & 0.800 & 0.798 & 0.798 \\
17 &  0.002 & 1.000 & 0.944 & 0.800 & 0.798 & 0.798
\end{tabular}
}
\caption{Consumer surplus and the market-clearing prices along iterated double auctions}\label{tab:num}
\end{table}

\section{Conclusions}

This paper analyzed the double auction mechanism in a pure exchange economy consisting of a finite set of agents with given initial endowments and preferences described by concave utility functions. The main results, Theorems~\ref{thm:ft}, \ref{thm:bdd} and \ref{thm:conv}, show that, under general conditions on the utility functions, repeated double auctions generate sequences of allocations and prices whose cluster points are Walras equilibria with transfers. To our knowledge, this is the first theoretical explanation of the price convergence observed in laboratory experiments with double auctions.

The presented theory has, however, some obvious limitations. First, both laboratory experiments and the computational results suggest that the allocations and prices generated by double auctions actually converge while Theorem~\ref{thm:conv} only asserts that all the cluster points are equilibria. Second, our results assume that the agents bid their indifference functions, but in reality, agents often behave more strategically by e.g.\ bid shading where they submit bids below their true indifference prices. Third, as discussed briefly in Remark~\ref{rem:cda}, most securities exchanges implement the double auction mechanism in continuous time, where incoming bids are matched immediately or added to the limit order book which is visible to the trading agents. Some of the above extensions will be addressed in a follow-up paper.

\section{Appendix}

This appendix gives a quick review of optimality conditions and duality in constrained convex optimization. All the results here can be found in \cite{roc70a} or \cite[Section~A.9]{pp24} and, in particular, \cite[Example~A.46]{pp24}.

Given integers $n,l,m$ with $l\le m$, let $f_j$, $j=0,\ldots,l$ be lsc and convex and $f_j$, $j=l+1,\ldots,m$ be affine functions on $\reals^n$ and consider the problem
\begin{equation}\label{p}
\begin{aligned}
  &\minimize & \quad f_0(x)&\quad \ovr x\in\reals^n,\\
  &\st & f_j(x)&\le 0\quad j=1,\ldots,l,\\
  & & f_j(x)&=0\quad j=l+1,\ldots,m.
\end{aligned}
\end{equation}
We will call this the {\em primal problem}. We assume that there exists at least one point $x\in\dom f_0$ that satisfies the constraints.

The {\em Lagrangian function} $L:\reals^n\times\reals^m\to\ereals$ associated with \eqref{p} is given by
\[
L(x,y)=
\begin{cases}
  +\infty & \text{if $x\notin\cap_{j=0}^l\dom f_j$,}\\
  f_0(x)+\sum_{j=1}^my_jf_j(x) & \text{if $x\in\cap_{j=0}^l\dom f_j$ and $y\in\reals^l_+\times\reals^{m-l}$,}\\
  -\infty & \text{otherwise}.
\end{cases}
\]
The {\em dual problem} is that of maximizing the concave function $g:\reals^m\to\ereals$ given by
\[
g(y):=\inf_{x\in\reals^n}L(x,y).
\]
A pair $(x,y)\in\reals^n\times\reals^m$ satisfies the {\em KKT-conditions} if 
\[
\partial_xL(x,y)\ni 0,\quad\partial_y[-L](x,y)\ni 0,
\]
where $\partial_x$ (resp.\ $\partial_y$) denotes the subdifferential with respect to the $x$-variable ($y$-variable).

Plugging in the definition of $L$ gives
\[
g(y) = \inf_{x\in\reals^n}\Big\{f_0(x)+\sum_{j=1}^my_jf_j(x)\;\Big|\; x\in\bigcap_{j=1}^l\dom f_j\Big\}
\]
for $y\in\reals^l_+\times\reals^{m-l}$ while $g(y):=-\infty$ for $y\notin\reals^l_+\times\reals^{m-l}$. The KKT conditions can be written as
\begin{equation*}
  \begin{gathered}
  \partial_x[f_0+\sum_{j=1}^my_jf_j](x)\ni 0,\\
  y_j\ge 0,\quad y_jf_j(x)=0\quad j=1,\ldots,l.
  \end{gathered}
\end{equation*}

The following is a special case of the equivalence (1) $\Leftrightarrow$ (5) in~\cite[Theorem~A.73]{pp24}.

\begin{theorem}\label{thm:kkt1}
A pair $(\bar x,\bar y)\in\reals^n\times\reals^m$ satisfies the KKT conditions if and only if $\bar x$ solves the primal, $\bar y$ solves the dual and the optimum values are equal.
\end{theorem}

The following can be found e.g.~in \cite[Example~A.76]{pp24}.

\begin{theorem}\label{thm:kkt2}
Assume that there exists a feasible $x^0\in\reals^n$ such that $f_j(x^0)<0$ for $j=1,\ldots, l$ and either $m=l$ (no equality constraints) or $x^0\in\cap_{j=0}^l\rinte\dom f_j$. Then the optimum values of the primal and the dual are equal and the dual optimum is attained. Moreover, an $\bar x\in\reals^n$ solves the primal problem if and only if it is feasible and there exists $\bar y\in\reals^m$ such that $(\bar x,\bar y)$ satisfies the KKT conditions.
\end{theorem}

\bibliographystyle{plain}
\bibliography{sp}

\end{document}